\documentclass[11.75pt]{article}
\usepackage[utf8]{inputenc}
\usepackage[left=2.5cm,right=2.5cm]{geometry}
\usepackage{amsthm,amsbsy,amssymb,amsmath, natbib}
\usepackage{url}
\usepackage{bbm}
\usepackage{algorithm}
\usepackage{algpseudocode}
\usepackage{caption}
\usepackage{subcaption}
\usepackage{color, xcolor}

\newcommand{\1}{\mathbbm{1}}
\newcommand{\E}{\mathbbm{E}}

\usepackage{tikz}
\usetikzlibrary{trees,automata,positioning}

\usepackage{algorithm}
\usepackage{algpseudocode}
\usepackage{caption}
\usepackage{subcaption}

\newtheorem{proposition}{Proposition}

\title{Properties of the generalized inverse Gaussian with applications to Monte Carlo simulation and distribution function evaluation}
\author{V\'ictor Pe\~na, \textit{Universitat Polit\`ecnica de Catalunya}\\
        Michael Jauch, \textit{Florida State University}}

\begin{document}

\maketitle
\begin{abstract}
The generalized inverse Gaussian, denoted $\mathrm{GIG}(p, a, b)$, is a flexible family of distributions that includes the gamma, inverse gamma, and inverse Gaussian distributions as special cases. In addition to its applications in statistical modeling and its theoretical interest, the GIG often arises in computational statistics, especially in Markov chain Monte Carlo (MCMC) algorithms for posterior inference. This article introduces two  mixture representations for the GIG: one that expresses the distribution as a continuous mixture of inverse Gaussians and another that reveals a recursive relationship between GIGs with different values of $p$. The former representation forms the basis for a data augmentation scheme that leads to a geometrically ergodic Gibbs sampler for the GIG. This simple Gibbs sampler, which alternates between gamma and inverse Gaussian conditional distributions, can be incorporated within an encompassing MCMC algorithm when simulation from a GIG is required. The latter representation leads to algorithms for exact, rejection-free sampling as well as CDF evaluation for the GIG with half-integer $p.$ 
\end{abstract}

\clearpage

\section{Introduction}

The generalized inverse Gaussian, which we denote  $\mathrm{GIG}(p, a, b)$, is a three-parameter, absolutely continuous distribution supported on $\mathbb{R}_+$. It is a rich family of distributions that encompasses the gamma distribution (if $b = 0$), the inverse gamma distribution (if $a = 0$), and the inverse Gaussian distribution (if $p = -1/2$). The parameters of the $\mathrm{GIG}(p, a, b)$ are $p \in \mathbb{R}$, $a > 0$, and $b > 0$, and its probability density function (PDF) is 
$$
f_X(x) = \frac{(a/b)^{p/2}}{2 K_p(\sqrt{ab})} \, x^{p-1} \exp\left\{  \frac{-(ax + b/x)}{2} \right\} \, \1(x > 0),
$$
where $K_p( \cdot)$ is the modified Bessel function of the second kind \citep{abramowitz1968handbook}. The moments of the GIG involve Bessel functions but are available in closed form (see, for example, \cite{jorgensen1982statistical}).

 \textcolor{black}{The generalized inverse Gaussian was first proposed by Halphen in the 1940s for analyzing hydrological data \citep{halphen1941nouveau, perreault1999halphen}. It regained popularity in the mid to late 1970s, when a series of articles studied its properties in detail (see, for example, \cite{barndorff1977infinite}, \cite{blaesild1978shape}, and \cite{halgreen1979self}). For an account of known properties and characterizations of the $\mathrm{GIG}(p, a, b)$, we refer the reader to the monograph \cite{jorgensen1982statistical} and the review article \cite{koudou2014characterizations}.}



The GIG often arises in computational statistics, especially in Markov chain Monte Carlo (MCMC) algorithms for posterior inference. For example, it emerges as a full conditional distribution for scale parameters in normal models. Sections~\ref{sec:gibbs} and~\ref{sec:half} describe or reference a variety of applications where it is necessary to simulate random draws from the GIG or evaluate its cumulative distribution function (CDF).


This article introduces two mixture representations of the generalized inverse Gaussian. One expresses the $\mathrm{GIG}(p, a, b)$ as a continuous mixture of inverse Gaussians, while the other reveals a relationship between GIGs. The former representation forms the basis for a data augmentation scheme that leads to a geometrically ergodic Gibbs sampler for the GIG. This Gibbs sampler, which alternates between gamma and inverse Gaussian conditional distributions, can be incorporated within an encompassing MCMC algorithm when simulation from a GIG is required. The latter representation leads to the first algorithm in the literature for exact, rejection-free sampling from the GIG with half-integer $p$ as well as an algorithm for CDF evaluation in that same setting. 

\color{black}

Several articles have been written on the topic of simulating from the GIG, including \cite{dagpunar1989easily}, \cite{hormann2014generating}, \cite{devroye2014random}, and \cite{zhang2022generator}. Recently, \cite{willmot2022remarks} derived expressions for the CDF of the generalized inverse Gaussian for half-integer $p$. In simulation experiments, we found that our implementations of the data-augmented Gibbs sampler and the exact simulation method can be more efficient than prominent \texttt{R} packages for simulating from the GIG, but this is not always the case. In our view, the methods proposed in this article complement rather than supersede the existing methods. 


The proofs of all propositions can be found in the supplementary material. In addition, the supplementary material includes pseudocode for the algorithms, a section discussing an alternative Gibbs sampler that can be derived from the results of \cite{zhang2022generator}, and a section describing how to adapt the propositions to an alternative parametrization of the GIG. The supplementary material also contains a discussion of the simulation experiments\footnote{Code  available at \url{https://github.com/michaeljauch/gig}.}.

\color{black}

 \subsection{Notation}

 We use the notation $\text{InvGauss}(\mu, \lambda)$ for the inverse Gaussian distribution with mean parameter $\mu$ and shape parameter $\lambda$. Similarly, $\text{Gamma}(\alpha, \beta)$ is the gamma distribution with shape parameter $\alpha$ and rate parameter $\beta$ and $\text{Exp}(\theta)$ is the exponential distribution with rate parameter $\theta$. 

\section{Data-augmented Gibbs sampler} \label{sec:gibbs}

In this section, we introduce a representation of the GIG$(p, a, b)$ as a continuous mixture of inverse Gaussian distributions. The representation is useful for deriving a data-augmented Gibbs sampler for the GIG. We then reference examples from the literature where the GIG appears as a full-conditional distribution in Gibbs sampling algorithms. An example shows how the Gibbs sampler can be nested within an encompassing MCMC algorithm when simulation from a GIG is required. Finally, we provide theoretical support for the Gibbs sampler: Proposition~\ref{prop:geo} establishes that it is geometrically ergodic.

\begin{proposition} \label{prop:mixt1}  The PDF of the $\mathrm{GIG}(p, a, b)$ for $p \neq -1/2$ can be written as a continuous mixture of inverse Gaussian distributions:
$$
f_X(x) = \int_{0}^\infty f_{X \mid Y}(x \mid y) f_{Y}(y) \, \mathrm{d} y,
$$
where $f_X(x)$ is the PDF of the $\mathrm{GIG}(p, a, b)$, $f_Y(y)$ is a PDF with support on $\mathbb{R}_+$, and $f_{X \mid Y}(x \mid y)$ is an inverse Gaussian PDF. 

The parameters of the distributions are different for the cases $p < -1/2$ and $p > -1/2$. If $p < -1/2$, then
\begin{align*}
f_{X \mid Y}(x) & \textcolor{black}{\, = f_{\mathrm{IG}}(x \, ; \, \sqrt{b/(a+2y)}, b)} \\
f_Y(y) &= \frac{1}{\mathcal{Z}_0} \, y^{-(p+3/2)} \, \exp\left\{ -\sqrt{b (a + 2y)} \right\} \\
\mathcal{Z}_0 &= \frac{b^{(p+1)/2} \, \sqrt{2} \, \Gamma(-p-1/2) \, K_p ( \sqrt{a b})}{a^{p/2}  \, \sqrt{\pi}},
\end{align*}
where $f_{\mathrm{IG}}(x \, ; \, \sqrt{b/(a+2y)}, b)$ is the PDF of the InvGauss$(\sqrt{b/(a+2y)}, b)$. If $p > -1/2$, then
\begin{align*}
f_{X \mid Y}(x)  & \,  \textcolor{black}{= f_\mathrm{IG}(x \, ; \, \sqrt{(b+2y)/a}, b + 2y)} \\
f_Y(y) &= \frac{1}{\mathcal{Z}_{1}} \, \frac{y^{p-1/2}}{\sqrt{b+2y}} \,  \exp \left\{ -\sqrt{a (b+2y)} \right\} \\
\mathcal{Z}_{1} &= \frac{b^{p/2} \, \sqrt{2} \, \Gamma(p+1/2) \, K_p ( \sqrt{a b})}{a^{p/2} \, \sqrt{\pi} }.
\end{align*}
\end{proposition}

To see why we have two separate cases in Proposition 1, notice that the PDF of the GIG can be written as 
$$
f_X(x) \propto x^{p+1/2} \,  \textcolor{black}{f_{\text{IG}}(x \, ; \, \sqrt{b/a}, b)}.
$$
We can substitute two different integral representations for $x^{p+1/2}:$ one that is valid for $p > -1/2$ and another one that is valid for $p < -1/2.$ If $p > -1/2$, we use the representation
$$
x^{p+1/2} \underset{(p > -1/2)}{=} \frac{\int_0^\infty y^{p-1/2} \, \exp( -y/x) \, \mathrm{d} y }{\Gamma(p+1/2)}.
$$
If $p < -1/2$, we use the representation
$$
x^{p+1/2} \underset{(p < -1/2)}{=} \frac{\int_0^\infty y^{-(p+3/2)} \, \exp \left( -x y \right) \, \mathrm{d} y}{\Gamma(-(p+1/2))}.
$$
The case $p = -1/2$ is not included in Proposition 1, but the GIG$(-1/2, a, b)$ is equivalent to $\mathrm{InvGauss}( \sqrt{b/a},  b)$ \citep{jorgensen1982statistical}. The conditional distribution $X \mid Y = y$ can be deduced from the implied joint PDF $f_{X,Y}(x,y)$. From the same bivariate PDF, we can also deduce that $Y \mid X = x$ is gamma-distributed.

\begin{proposition} \label{prop:cond}
    Let $f_{X \mid Y}(x)$ and $f_Y(y)$ be as defined in Proposition 1. If $p < -1/2$, then $Y \mid X = x \sim \mathrm{Gamma}(-(p+1/2), x)$. If $p > -1/2$, then $Y \mid X = x \sim \mathrm{Gamma}(p+1/2, 1/x)$.
\end{proposition}

One can define a data-augmented Gibbs sampler for the GIG by sampling iteratively from the conditional distributions $X \mid Y = y$ and $Y \mid X = x$. Only the draws from $X \mid Y = y$ are relevant if we are interested in sampling from the GIG$(p,a,b)$. 

\textcolor{black}{In Propositions 1 and 2, we treat the cases $p< -1/2$ and $p> -1/2$ separately to obtain general results for all $p \neq -1/2$. For building a Gibbs sampler, one can focus on one of the cases and leverage the fact that if $X \sim \mathrm{GIG}(-p,b,a)$, then $1/X \sim \mathrm{GIG}(p,a,b)$ \citep{jorgensen1982statistical}.}

The Gibbs sampler requires sampling from inverse Gaussian and gamma distributions. There are efficient algorithms to generate random draws from these distributions. For example, one can generate inverse Gaussian draws with the algorithm described in \cite{michael1976generating}, whereas for the gamma distribution, one can use standard algorithms such as those proposed in \cite{ahrens1974computer}.

\color{black} 

The Gibbs sampler is geometrically ergodic. Geometric ergodicity is a property that quantifies the rate of convergence of the Gibbs sampler to its stationary distribution (see, for example, \cite{johnson2015geometric} for further details). It also ensures that inferences made with draws from the chain are well-behaved. In particular, if $\{X_i\}_{i=1}^n$ are draws from the chain and we want to estimate a function $f(X)$ with finite $\mathbb{E}[ f(X)^{2+\varepsilon}]$ for some $\varepsilon > 0$, geometric ergodicity ensures that $\sum_{i = 1}^n f(X_i)/n$ is approximately normal \citep{chan1994discussion}. Our proof of geometric ergodicity relies on Theorem 3.5 in \cite{johnson2015geometric}.

\begin{proposition} \label{prop:geo} The data-augmented composition Gibbs sampler implied by Propositions~\ref{prop:mixt1} and~\ref{prop:cond} is geometrically ergodic.
\end{proposition}

\color{black}

The GIG$(p,a,b)$ often arises as a full conditional distribution in Gibbs sampling algorithms. There are countless examples from the literature on global-local shrinkage priors \citep{armagan2011}, graphical models \citep{khare2018}, stochastic volatility models \citep{barndorffnielsen1997}, Bayesian nonparametrics \citep{favaro2012}, smoothing splines \citep{loyal2024}, and other areas. The Gibbs sampler can be embedded within an encompassing MCMC algorithm when simulation from a GIG$(p,a,b)$ is required. The following example offers a stylized illustration. 

\textbf{Example:}
\textit{Suppose that we observe independent and identically distributed data $y_i \mid \mu, \sigma^2 \sim \mathrm{Normal}\left(\mu, \sigma^2 \right)$, $i \in \{1, 2, \, ... \, , n\}$, with $\mu$ and  $\sigma^2$ unknown. We assign independent priors to $\mu$ and  $\sigma^2$ with $\mu \sim \mathrm{Normal}\left(\theta_0, \tau_0^2\right)$ and $\sigma^2 \sim \mathrm{GIG}\left(p_0, a_0, b_0\right).$ Letting $\mathbf{y} = \left(y_1, \ldots, y_n\right)^\top$ and $\bar{y} = \frac{1}{n}\sum_{i=1}^n y_i,$ we can simulate from the posterior distribution of $\left(\mu, \sigma^2\right)$ by iterating between the full conditional distributions
\begin{align*}
\mu \mid \sigma^2, \mathbf{y} &\sim \mathrm{Normal}\left( \theta_n, \tau_n^2 \right) \\ 
\sigma^2 \mid \mu, \mathbf{y} &\sim \mathrm{GIG}\left(p_n, a_n, b_n \right),
\end{align*}
where
\begin{align*} 
\theta_n = \left(\frac{n}{\sigma^2} + \frac{1}{\tau_0^2}\right)^{-1} \left(\frac{n}{\sigma^2}\bar{y} + \frac{1}{\tau_0^2}\theta_0 \right),  \quad 
\tau_n^2 = \left(\frac{n}{\sigma^2} + \frac{1}{\tau_0^2}\right)^{-1},
\end{align*} and
\begin{align*} 
p_n = p_0 - \frac{n}{2},  \quad a_n = a_0, \quad b_n = b_0 + \sum_{i=1}^n (y_i - \mu)^2.
\end{align*}
\textit{We can leverage the results of this section to derive a Gibbs sampler that does not require simulating directly from the GIG$(p,a,b)$. If $p_n < -1/2,$ we can introduce an auxiliary variable $\omega$ and iterate through the full conditional distributions}
\begin{align*}
\mu \mid \sigma^2, \mathbf{y} &\sim \mathrm{Normal}\left( \theta_n, \tau_n^2 \right) \\ 
\sigma^2 \mid \mu, \omega, \mathbf{y} &\sim \mathrm{InvGauss}\left(\sqrt{b_n/(a_n + 2\omega)}, b_n \right) \\ 
\omega \mid \mu, \sigma^2 , \mathbf{y} &\sim \mathrm{Gamma}\left(-(p_n + 1/2), \sigma^2 \right).
\end{align*} \textit{If $p_n > -1/2,$ the full conditional distributions are}
\begin{align*}
\mu \mid \sigma^2, \mathbf{y} &\sim \mathrm{Normal}\left( \theta_n, \tau_n^2 \right) \\ 
\sigma^2 \mid \mu, \omega, \mathbf{y} &\sim \mathrm{InvGauss}\left(\sqrt{(b_n + 2\omega)/a_n}, b_n + 2\omega \right) \\ 
\omega \mid \mu, \sigma^2 , \mathbf{y} &\sim \mathrm{Gamma}\left(p_n + 1/2, 1/\sigma^2 \right).
\end{align*}}

\textcolor{black}{
We conducted a simulation experiment to compare the efficiency of the two Gibbs samplers described in the example. We found that the data-augmented Gibbs sampler was more than 9 times faster than the original Gibbs sampler when the latter was implemented with the method of \citet{devroye2014random} from the R package \texttt{boodist} \citep{boodist}. However, the data-augmented Gibbs sampler was approximately 8\% slower when the original Gibbs sampler was implemented with the method of \citet{hormann2014generating} from the R package \texttt{GIGrvg} \citep{GIGrvg}. Neither the original Gibbs sampler nor the data-augmented Gibbs sampler produce independent draws from the posterior distribution, so it is important to evaluate the effective sample size (ESS) in addition to the time it takes to produce the Markov chains. For the original Gibbs sampler, we found that the ESS was close to the nominal sample size for both parameters $\mu$ and $\sigma^2.$ For the data-augmented Gibbs sampler, this was only true for $\mu.$ The ESS of $\sigma^2$ was about 1/3 of the nominal sample size. The details of the simulation experiment can be found in the supplementary material.}

\section{Exact sampling and CDF evaluation for half-integer $p$} \label{sec:half}


The main results of this section reveal a relationship between GIGs with different values of $p.$ 
Proposition \ref{prop:mixt2} leads to the first algorithm in the literature for exact, rejection-free simulation from the GIG$(p,a,b)$ for half-integer $p.$ Proposition \ref{prop:CDF} leads to an algorithm for evaluating the CDF of the GIG$(p,a,b)$ for half-integer $p.$
{\textcolor{black}{Pseudocode for the algorithms can be found in the supplementary material.}}

{\begin{proposition} \label{prop:mixt2} If $X \sim \mathrm{GIG}(p, a, b)$ and $p >1$, then $X =_d Y + E $ for independent random variables $Y$ and $E$ with
\begin{align*}
f_{Y}(y) &=  \textcolor{black}{w \, f_{\mathrm{GIG}}(y \, ; \, p-2, a, b) + (1 - w)  \, f_{\mathrm{GIG}}(y \, ; \, p-1, a, b)} \\ 
w &= {K_{p-2}(\sqrt{ab})}/{K_p(\sqrt{ab})} \in (0, 1)
\end{align*}
and $E \sim \mathrm{Exp}(a/2)$, \textcolor{black}{where $f_{\mathrm{GIG}}(y \, ; \, p, a, b)$ is the PDF of the GIG$(p,a,b)$.}
\end{proposition}}  
\color{black}

The proposition above can be proven with Theorem 2.4 of \cite{jauch2021mixture}. It shows that if $p > 1$ is half-integer, one can simulate from the $\mathrm{GIG}(p, a, b)$ by recursively simulating GIGs of smaller order $p$ until reaching the base case $p = 3/2$. In that case, the components of the mixture are $\mathrm{GIG}(-1/2, a, b)$, which is equivalent to $\mathrm{InvGauss}(\sqrt{b/a}, b)$. The case $\mathrm{GIG}(1/2, a, b)$ can be handled with an inverse Gaussian as well because if $X \sim \mathrm{GIG}(-p, b, a)$, then $1/X \sim \mathrm{GIG}(p, a, b)$. This property can also be used to include the cases where $p$ is a negative half integer.

The exact simulation method described above should be very efficient when $|p|$ is sufficiently small. To verify this, we conducted an experiment comparing the exact simulation method with the method \citet{devroye2014random} from the R package \texttt{boodist} \citep{boodist} and the method of  \citet{hormann2014generating} from the R package \texttt{GIGrvg} \citep{GIGrvg} in the case where $p=3/2.$ As expected, the exact simulation method outperformed the other two methods \textcolor{black}{in terms of run time.} The details of the simulation experiment can be found in the supplementary material.


\begin{proposition} \label{prop:CDF} Let $G_p(x)$ be the CDF of the $\mathrm{GIG}(p,a,b)$. For half-integer $p$ such that $|p| > 1/2$, the following recurrence holds:
$$
G_p(x) = w G_{p-2}(x) + (1-w) G_{p-1}(x) - \exp\left\{ \frac{-ax}{2} \right\} I_p(x),
$$
where
$$
I_p(x) = \frac{w (ab)^{(p-2)/2}}{2^{p-1} K_{p-2}(\sqrt{ab})} \Gamma\left(2-p, \frac{b}{2x}\right) + \frac{(1-w) (ab)^{(p-1)/2}}{2^p K_{p-1}(\sqrt{ab})} \Gamma\left(1-p, \frac{b}{2x}\right)
$$
and $\Gamma( \cdot,  \cdot)$ is the incomplete gamma function \citep{abramowitz1968handbook}.
\end{proposition}

The recurrence formula for $G_p$ implies a recursive algorithm for evaluating the CDF of the GIG$(p, a, b)$ for half-integer $p$. The base case involves the CDFs $G_{-1/2}$ and $G_{1/2}$, which are based on inverse Gaussian distributions. 


\color{black}
 
It is not uncommon that statistical methods require sampling from or evaluating the CDF of the $\text{GIG}(p, a, b)$ for half-integer $p.$ \citet{favaro2012} introduce a stick-breaking representation for the normalized inverse Gaussian process \citep{lijoi2005} whose practical value depends upon being able to efficiently simulate from the $\text{GIG}(p, a, b)$ for half-integer $p.$ In both \citet{bhattacharya2015} and \citet{loyal2024}, GIG$(p,a,b)$ full conditional distributions have half-integer $p$ parameters for the priors considered. \citet{he2022} present a data augmentation scheme for models that include gamma functions that requires simulating from a sum of independent GIGs for which $p=-3/2.$ \citet{willmot2022remarks} describe several financial and actuarial applications in which one needs to evaluate the CDF of the $\text{GIG}(p, a, b)$ for half-integer $p.$

\color{black}




\section{Conclusions} \label{sec:conc}

This article introduced two mixture representations of the GIG: one as a continuous mixture of inverse Gaussian distributions, and the other as a mixture of two GIGs plus an exponential random variable. These mathematical results led to conceptually straightforward algorithms for Monte Carlo simulation and CDF evaluation. As future work, it would be interesting to see if there are analogous representations for the matrix-variate version of the generalized inverse Gaussian \citep{barndorff1982exponential, hamura2023gibbs}, a probability distribution over symmetric positive-definite matrices that includes the Wishart and inverse Wishart as special cases.

\color{black}



\clearpage
\bibliography{GIGs}
\clearpage

\appendix

\section{Appendix}

\textcolor{black}{Section 1 includes the proofs for the propositions stated in the main text. Section 2 contains pseudocode for the algorithms. Section 3 considers an alternative parametrization of the GIG. Section 4 includes a derivation of a Gibbs sampler that can be easily deduced from results in \cite{zhang2022generator}. Finally, Section 5 includes a discussion on results from numerical experiments.}

\subsection{Proofs of Propositions}

\newtheorem{prop}{Proposition}

\begin{prop} The PDF of the $\mathrm{GIG}(p, a, b)$ for $p \neq -1/2$ can be written as a continuous mixture of inverse Gaussian distributions:
$$
f_X(x) = \int_{0}^\infty f_{X \mid Y}(x \mid y) f_{Y}(y) \, \mathrm{d} y,
$$
where $f_X(x)$ is the PDF of the $\mathrm{GIG}(p, a, b)$, $f_Y(y)$ is a PDF with support on $\mathbb{R}_+$, and $f_{X \mid Y}(x \mid y)$ is an inverse Gaussian PDF. 

The parameters of the distributions are different for the cases $p < -1/2$ and $p > -1/2$. If $p < -1/2$, then
\begin{align*}
f_{X \mid Y}(x) & \textcolor{black}{\, = f_{\mathrm{IG}}(x \, ; \, \sqrt{b/(a+2y)}, b)} \\
f_Y(y) &= \frac{1}{\mathcal{Z}_0} \, y^{-(p+3/2)} \, \exp\left\{ -\sqrt{b (a + 2y)} \right\} \\
\mathcal{Z}_0 &= \frac{b^{(p+1)/2} \, \sqrt{2} \, \Gamma(-p-1/2) \, K_p ( \sqrt{a b})}{a^{p/2}  \, \sqrt{\pi}},
\end{align*}
where $f_{\mathrm{IG}}(x \, ; \, \sqrt{b/(a+2y)}, b)$ is the PDF of the InvGauss$(\sqrt{b/(a+2y)}, b)$. If $p > -1/2$, then
\begin{align*}
f_{X \mid Y}(x)  & \,  \textcolor{black}{= f_\mathrm{IG}(x \, ; \, \sqrt{(b+2y)/a}, b + 2y)} \\
f_Y(y) &= \frac{1}{\mathcal{Z}_{1}} \, \frac{y^{p-1/2}}{\sqrt{b+2y}} \,  \exp \left\{ -\sqrt{a (b+2y)} \right\} \\
\mathcal{Z}_{1} &= \frac{b^{p/2} \, \sqrt{2} \, \Gamma(p+1/2) \, K_p ( \sqrt{a b})}{a^{p/2} \, \sqrt{\pi} }.
\end{align*}
\end{prop}

\begin{proof} First, we find $f_{X\mid Y}(x \mid y)$. Along the way, we derive $f_{Y \mid X}(y \mid x)$, which will be useful for deriving a Gibbs sampler. We consider two cases: $p < -1/2$ and $p > -1/2$.

\textbf{Case $p < -1/2$:} Up to proportionality constants, the PDF $f_X(x)$ of the $\mathrm{GIG}(p, a, b)$ is
$$
f_X(x) \propto x^{p+1/2} \, f_{\mathrm{IG}}(x \, ; \,\sqrt{b/a},  b),
$$
{\textcolor{black}{where $f_{\mathrm{IG}}(x \, ; \,\mu,  \lambda)$ is the PDF of an inverse Gaussian:
$$
f_{\mathrm{IG}}(x \, ; \,\mu, \lambda) = \sqrt{\frac{\lambda}{2 \pi x^3}} \, \exp\left\{-\frac{\lambda (x - \mu)^2}{2 \mu^2 x}\right\} \, \mathbbm{1}(x > 0),
$$
for $\mu > 0$ and $\lambda > 0$.}

The term $x^{p+1/2}$ can be represented as follows
$$
x^{p+1/2} \underset{(p < -1/2)}{=} \frac{\int_0^\infty y^{-(p+3/2)} \, \exp \left( -x y \right) \, \mathrm{d} y}{\Gamma(-(p+1/2))} \, ,
$$
so we can write
$$
f_X(x) \propto f_{\mathrm{IG}}(x \, ; \,\sqrt{b/a},  b) \, \int_0^\infty y^{-(p+3/2)} \, \exp \left( -x y \right) \, \mathrm{d} y .
$$
If we define
$$
g(x,y) =  f_{\mathrm{IG}}(x \, ; \,\sqrt{b/a},  b) \, y^{-(p+3/2)} \, \exp \left( -x y \right)  \mathbbm{1}(x >0, y >0),
$$
then $g(x,y)$ is non-negative and continuous. Therefore, by Tonelli's theorem, 
the double integral on $\mathbbm{R}^2_+$ of $g(x,y)$ is equal to the iterated integral, which is finite. This implies that $g(x,y)$ can be normalized to be a joint PDF $f_{X,Y}(x, y) \propto g(x,y)$ whose marginal in $X$ is $f_X(x)$. By inspecting $g(x,y)$, we can easily deduce that 
\begin{align*}
    X \mid Y = y &\sim \mathrm{InvGauss}(\sqrt{b/(a+2y)}, b) \\
    Y \mid X = x &\sim \mathrm{Gamma}(-(p+1/2), x).
\end{align*}

\textbf{Case $p > -1/2$:} We use the same argument for $p < -1/2$ but with a different integral representation for $x^{p+1/2}$. In this case, we use the representation
$$
x^{p+1/2} \underset{(p > -1/2)}{=} \frac{\int_0^\infty y^{p-1/2} \, \exp( -y/x) \, \mathrm{d} y }{\Gamma(p+1/2)}.
$$
We can rewrite the PDF of the $\mathrm{GIG}(p, a, b)$ as
$$
f_X(x)  \propto  f_{\mathrm{IG}}(x \, ; \,\sqrt{b/a},  b) \, \int_0^\infty y^{p-1/2} \, \exp \left( -y/x \right) \, \mathrm{d} y.
$$
By Tonelli's theorem, the joint PDF of $(X,Y)$ is proportional to
$$
f_{X,Y}(x, y) \propto y^{p-1/2} \exp \left\{ - \sqrt{a (b+2 y)} \right\} \,  f_{\mathrm{IG}}(x \, ; \,\sqrt{(b+2y)/a}, b+2y).
$$
The conditional distributions are
\begin{align*}
    X \mid Y = y &\sim \mathrm{InvGauss}(\sqrt{(b+2y)/a}, b+2y) \\
    Y \mid X = x &\sim \mathrm{Gamma}(p+1/2, 1/x).
\end{align*}
Lastly, we find $f_Y(y)$. Again, we treat the cases $p < -1/2$ and $p > -1/2$ separately. Let $f_{\mathrm{G}}(x \, ; \, \alpha, \beta)$ be the PDF of the gamma distribution:
$$
f_{\mathrm{G}}(x \, ; \, \alpha, \beta) = \frac{\beta^\alpha}{\Gamma(\alpha)} y^{\alpha-1} \exp\{-\beta y\} \, \mathbbm{1}(y > 0) 
$$
for $\alpha > 0$ and $y > 0$. If $p < - 1/2$, then:
\begin{align*}
f_Y(y)^{-1} &= \int_{0}^\infty \frac{f_X(x \mid  y)}{f_Y(y \mid x)} \, \mathrm{d} x \\
&=  \int_{0}^\infty \frac{f_{\mathrm{IG}}(x \, ; \,\sqrt{b/(a+2y)}, b) }{ f_{\mathrm{G}}(x \, ; \, - (p +1/2), x) } \, \mathrm{d} x.
\end{align*}
Before we proceed with the integral, we write the explicit definitions of the densities. The inverse Gaussian density is
\begin{align*}
f_{\mathrm{IG}}(x \, ; \,\sqrt{b/(a+2y)}, b) &= \frac{b^{1/2} x^{-3/2}}{(2\pi)^{1/2}} \, \exp\left\{ \frac{-(a+2y)}{2x} \left[x^2 - 2x \sqrt{b/(a+2y)} + b/(a+2y) \right] \right\} \\
&= \frac{b^{1/2} x^{-3/2}}{(2\pi)^{1/2}} \exp\left\{ \, \frac{-(a+2y) x}{2} - \frac{b}{2x} + \sqrt{(a+2y)b} \right\}.
\end{align*}
The gamma density is
$$
f_G(y \, ; \, -(p+1/2), x) = \frac{x^{-p-1/2}}{\Gamma(- p -1/2)} y^{-p-3/2} \exp\{-yx\}.
$$
Therefore,
\begin{align*}
f_Y(y)^{-1} &=  \int_{0}^\infty \frac{f_{\mathrm{IG}}(x \, ; \,\sqrt{b/(a+2y)}, b) }{ f_{\mathrm{G}}(x \, ; \, - (p +1/2), x) } \, \mathrm{d} x \\
&= \mathcal{K}_0 \int_0^\infty \underbrace{x^{p-1} \exp\left\{ \frac{-(ax + b/x)}{2} \right\}}_{\propto \mathrm{GIG}(p,a,b)} \, \mathrm{d} x \\
&= \mathcal{K}_0 \frac{2 K_p(\sqrt{ab})}{(a/b)^{p/2}}
\end{align*}
where
$$
\mathcal{K}_0 = \frac{b^{1/2} \Gamma(-p-1/2) y^{p+3/2} \exp\{\sqrt{b(a+2y)}\}}{(2\pi)^{1/2}} .
$$
Rearranging normalizing constants, we obtain
\begin{align*}
f_Y(y) &= \frac{1}{\mathcal{Z}_0} \, y^{-(p+3/2)} \, \exp\left\{ -\sqrt{b (a + 2y)} \right\} \\
\mathcal{Z}_0 &= \frac{b^{(p+1)/2} \, \sqrt{2} \, \Gamma(-p-1/2) \, K_p ( \sqrt{a b})}{a^{p/2}  \, \sqrt{\pi}}. 
\end{align*}}
\noindent The case $p > -1/2$ can be handled similarly:
\begin{align*}
f_Y(y)^{-1} &= \int_{0}^\infty \frac{f_X(x \mid  y)}{f_Y(y \mid x)} \, \mathrm{d} x \\
&=  \int_{0}^\infty \frac{f_{\mathrm{IG}}(x \, ; \,\sqrt{(b+2y)/a}, b+2y) }{ f_{\mathrm{G}}(x \, ; \, p +1/2, 1/x) } \, \mathrm{d} x \\
&= \mathcal{Z}_{1} \, y^{-(p-1/2)} \, (b+2y)^{1/2}    \, \exp \left\{ a^{1/2}(b+2y)^{1/2} \right\},
\end{align*}
where
$$
\mathcal{Z}_{1} = \frac{2^{1/2} \Gamma(p+1/2) K_p\left( a^{1/2} b^{1/2} \right)}{\pi^{1/2} (a/b)^{p/2}}.
$$
\end{proof}

\begin{prop}
    Let $f_{X \mid Y}(x)$ and $f_Y(y)$ be as defined in Proposition 1. If $p < -1/2$, then $Y \mid X = x \sim \mathrm{Gamma}(-(p+1/2), x)$. If $p > -1/2$, then $Y \mid X = x \sim \mathrm{Gamma}(p+1/2, 1/x)$.
\end{prop}

\begin{proof}
    We found these distributions in Proposition 1. They can be derived after inspecting the joint distributions (up to proportionality constants). If $p < -1/2$, then 
$$
   f_{X,Y}(x,y) \propto f_{\mathrm{IG}}(x \, ; \,\sqrt{b/a},  b) \, y^{-(p+3/2)} \, \exp \left( -x y \right),
$$
from which it is straightforward to deduce that
\begin{align*}
    Y \mid X = x &\sim \mathrm{Gamma}(-(p+1/2), x).
\end{align*}
If $p > -1/2$, then
$$
f_{X,Y}(x, y) \propto y^{p-1/2} \exp \left\{ - \sqrt{a (b+2 y)} \right\} \,  f_{\mathrm{IG}}(x \, ; \,\sqrt{(b+2y)/a}, b+2y),
$$
which implies that
\begin{align*}
    Y \mid X = x &\sim \mathrm{Gamma}(p+1/2, 1/x).
\end{align*}
\end{proof}

\begin{prop} The data-augmented composition Gibbs sampler in Algorithm 1 is geometrically ergodic.
\end{prop}

\begin{proof} We prove this result by checking that the conditions of Theorem 3.5 in \cite{johnson2015geometric} are satisfied. The proof amounts to showing that there exist functions $f, g: [0, \infty) \rightarrow [1, \infty)$ and constants $j, k, m, n > 0$ with $jm < 1$ so that
$$
\E[f(X) \mid Y = y] \le j g(y) + k, \qquad \E[g(Y) \mid X = x] \le m f(x) + n,
$$
with $C_d = \{y : g(y) \le d\}$ compact for all $d > 0$. We prove the results for $p < -1/2$ and $p > 1/2$ separately.

\textbf{Case $p < -1/2$:} Define
\begin{align*}
\theta &= 3/(2b) \\
\mu &= 3(1/2-p)/(2b) \\
f(x) &= 1 + \alpha (1/x - \theta)^2, \qquad 0 < \alpha < b^2/(p^2-1/4) \\
g(y) &= 1 + \beta (y-\mu)^2, \qquad \alpha^2/b^2 < \beta < \alpha/(p^2-1/4).
\end{align*}
The set $\{y : g(y) \le d \}$ is compact for all $d > 0$ and $f, g: [0, \infty) \rightarrow [1, \infty)$.

We find the conditional expectations of $f$ and $g$ and bound them. We start with the conditional expectation of $f$:
\begin{align*}
    \E[f(X) \mid Y = y] &= 1 + \alpha \left\{ \mathrm{Var}(1/X \mid Y = y) + \left[\E(1/X \mid Y = y) - \theta \right]^2 \right\}.
\end{align*}
Using well-known properties of the inverse Gaussian distribution, we find
\begin{align*}
    \mathrm{Var}(1/X \mid Y = y) &= {(a+2y)^{1/2}}{b^{-3/2}} + {2}{b^{-2}} \\
    \E(1/X \mid Y = y) - \theta &= (a+2y)^{1/2}{b}^{-1/2} + b^{-1} - (3/2) b^{-1}.
\end{align*}
Rearranging terms, we obtain
\begin{align*} 
    \E[f(X) \mid Y = y] &= 1 + 9\alpha/(4b^2)  + \alpha a/b + 2 \alpha y/b \\
    &= 1 + 9\alpha/(4b^2)  + \alpha a/b + 2 \alpha \mu/b + 2 \alpha (y-\mu)/b \\
    &\le 1 + 9\alpha/(4b^2)  + \alpha a/b + 2 \alpha \mu/b + g(y),
\end{align*}
where the last inequality holds because $\beta \ge \alpha^2/b^2$. The inequality can be rewritten as
\begin{align*}
\E[f(X) \mid Y = y] &\le j g(y) + k \\
j &= 1 \\
k &= 1 + 9\alpha/(4b^2)  + \alpha a/b + 2 \alpha \mu/b.
\end{align*}
Both $j$ and $k$ are positive, which is required by Theorem 3.5 in \cite{johnson2015geometric}. Now, we find the conditional expectation of $g$ and bound it:
\begin{align*}
 \E[g(Y) \mid X = x] &= 1 + \beta \left\{ \mathrm{Var}(Y \mid X = x) + \left[\E(Y \mid X =x) - \mu \right]^2 \right\} \\
 &= 1 + \beta \left\{ -(p+1/2) x^{-2} + \left[ -(p+1/2)x^{-1} - \mu \right]^2 \right\} \\
 &= 1 + \beta (1/2-p)/(4b^2) + {\beta (p^2-1/4)}  \left[1/x - \theta \right]^2 \\
 &\le 1 + 9 \beta (1/2-p)/(4b^2)   + \beta (p^2-1/4) f(x)/\alpha.
\end{align*}
We can rewrite the inequality as
\begin{align*}
 \E[g(Y) \mid X = x] &\le m f(x) + n \\
 m &=   \beta (p^2-1/4)/\alpha \\
 n &= 1 + 9 \beta (1/2-p)/(4b^2),
\end{align*}
where $m, n > 0$, as required by Theorem 3.5 in \cite{johnson2015geometric}. It remains to show that $j m < 1$. Given our choices of $\alpha$ and $\beta$,
$$
0 < {\beta (p^2-1/4)}/{\alpha}  < 1
$$
and the interval $(\alpha^2/b, \alpha/(p^2-1/4))$ is nonempty. This completes the proof for this case.

\textbf{Case $p > -1/2:$} Let $0 < \gamma < 1$ and
\begin{align*}
\theta &= 1/(2a) \\
\mu &= (p+1/2)/(2a) \\
f(x) &= 1 + \alpha (x - \theta)^2, \qquad \alpha > 0 \\
g(y) &= 1 + \beta(y-\mu)^2, \qquad \beta > \max[\alpha^2/(\gamma^2 a), {\alpha}/{(p+1/2)^2}].
\end{align*}
Again, $\{y : g(y) \le d \}$ is compact for all $d > 0$ and $f,g: [0, \infty) \rightarrow [1, \infty)$. We bound the conditional expectations. We start with $f$:
\begin{align*}
\E[f(X) \mid Y = y] &= 1 + \alpha \left\{ \mathrm{Var}(X \mid Y = y) +  \left[ \E(X \mid Y = y)  - \theta \right]^2 \right\} \\
&= 1 + \alpha \left\{  {(b+2y)^{1/2}}{a^{-3/2}} +  \left[ (b+2y)^{1/2} a^{-1/2}   - \theta \right]^2 \right\} \\
&= 1 +   \alpha /(4a^2)+ \alpha b /a  +  {2 \alpha}\mu/a + {2 \alpha} (y-\mu)/a \\
&\le  1 +   \alpha /(4a^2)+ \alpha b /a  +  {2 \alpha}\mu/a + \gamma g(y).
\end{align*}
The inequality holds because $\beta > \alpha^2/(\gamma^2 a^2)$. Therefore,
\begin{align*}
\E[f(X) \mid Y = y]  &\le j g(y) + k \\
j &= \gamma \\
k &= 1 +   \alpha /(4a^2)+ \alpha b /a  +  {2 \alpha}\mu/a,
\end{align*}
where $j$ and $k$ are positive, as required by Theorem 3.5 in \cite{johnson2015geometric}. Finally, we bound the conditional expectation of $g$:
\begin{align*}
    \E[g(Y) \mid X = x] &= 1 + \beta \left\{ \mathrm{Var}(Y \mid X = x) + [\E(Y \mid X = x) - \mu]^2\right\} \\
    &= 1 + \beta \left\{ (p+1/2)x^2 + [(p+1/2)x - \mu]^2\right\} \\
    &= 1 + \beta (p+1/2)x^2 + \beta (p+1/2)^2 (x - \theta)^2 \\
    &= \beta (p+1/2)x^2 + [\beta (p+1/2)^2 - \alpha](x-\theta)^2 + f(x) \\
    &= Q(x) + f(x) \\
    &\le c_0 + f(x),
\end{align*}
where $c_0 > 0$ because the quadratic $Q(x)$ is strictly positive. We can rewrite the inequality as
\begin{align*}
    \E[g(Y) \mid X = x] &\le m f(x) + n \\
    m &= 1 \\
    n &= c_0,
\end{align*}
where $m, n > 0$. It remains to show that $j m < 1$. Since $0 < \gamma < 1$ by assumption, the proof is now complete.
\end{proof}

{\begin{prop}  If $X \sim \mathrm{GIG}(p, a, b)$ and $p >1$, then $X =_d Y + E $ for independent random variables $Y$ and $E$ with
\begin{align*}
f_{Y}(y) &=  \textcolor{black}{w \, f_{\mathrm{GIG}}(y \, ; \, p-2, a, b) + (1 - w)  \, f_{\mathrm{GIG}}(y \, ; \, p-1, a, b)} \\ 
w &= {K_{p-2}(\sqrt{ab})}/{K_p(\sqrt{ab})} \in (0, 1)
\end{align*}
and $E \sim \mathrm{Exp}(a/2)$, \textcolor{black}{where $f_{\mathrm{GIG}}(y \, ; \, p, a, b)$ is the PDF of the GIG$(p,a,b)$.}
\end{prop}}  
\color{black}

\begin{proof}
First, we prove that the PDF of the generalized inverse Gaussian distribution for $p > 1$ can be written as a continuous mixture of truncated exponential random variables. After that, it will be straightforward to prove the main result. More precisely, we show that
\begin{align*}
f_X(x) &= \int_0^\infty  f_{X \mid Y}(x \mid y) f_Y (y) \, \mathrm{d} y  \\ &= \int_0^\infty \frac{a}{2} \exp\left\{ - \frac{a}{2} (x - y) \right\} \mathbbm{1}(x \ge y) \, f_Y(y) \, \mathrm{d} y,
\end{align*}
where  
\begin{align*}
f_{Y}(y) &= w \, f_{\mathrm{GIG}}(y \, ; \, p-2, a, b) + (1 - w)  \, f_{\mathrm{GIG}}(y \, ; \, p-1, a, b) \\
w &= \frac{K_{p-2}(\sqrt{ab})}{K_p(\sqrt{ab})}.
\end{align*}
Let $a > 0, b > 0, p > 1$, and define:
\begin{align*}
f(x) &=  \frac{a}{2} \exp\left(- \frac{a}{2} x \right) \, \mathbbm{1}(x > 0) \\
\qquad g(x) &= \frac{(a/b)^{p/2}}{2 K_p(\sqrt{ab})} x^{p-1} \exp\left\{- \frac{1}{2} (a x + b/x) \right\} \, \mathbbm{1}(x > 0). 
\end{align*}
The ratio 
$$
\frac{g(x)}{f(x)} = \frac{a^{p/2-1}}{b^{p/2} K_p(\sqrt{ab})} \, x^{p-1} \, \exp\left( - \frac{b}{2x} \right) \,  \mathbbm{1}(x > 0)
$$
is monotone non-decreasing in $x > 0$ for $p > 1$, so we can apply Theorem 2.4 in \cite{jauch2021mixture} to write $g(x)$ as a mixture of truncated $\mathrm{Exp}(a/2)$ random variables:
$$
g(x)  =  \int_0^\infty \frac{a}{2} \exp\left\{ - \frac{a}{2} (x - y) \right\} \mathbbm{1}(x \ge y) \, f_Y(y) \, \mathrm{d} y.
$$
By Theorem 2.4 in \cite{jauch2021mixture}, the mixing density $f_Y(y)$ is
\begin{align*}
f_Y(y) &=  \left(\frac{g(y)}{f(y)} \right)' \, \int_y^\infty f(s) \, \mathrm{d} s   \\
&=  \frac{a^{p/2-1}}{b^{p/2} K_p(\sqrt{ab})} \, \left[\frac{b}{2} y^{p-3} + (p-1) y^{p-2}  \right] \exp\left\{- \frac{1}{2} (ay + b/y) \right\}  \, \mathbbm{1}(y > 0) \\
&=   w \, f_{\mathrm{GIG}}(y \, ; \, p-2, a, b) + (1 - w)  \, f_{\mathrm{GIG}}(y \, ; \, p-1, a, b) .
\end{align*}
Now, we prove the main result. Let $\tilde{X} = Y + E$. It remains to show that $\tilde{X} =_d X$. By the convolution formula,
\begin{align*}
f_{\tilde{X}}(x) &= \int_0^\infty f_E(x-y) f_Y(y)  \, \mathrm{d} y \\
&= \int_0^\infty \frac{a}{2} \exp\left\{ - \frac{a}{2} (x - y) \right\} \mathbbm{1}(x \ge y) \, f_Y(y) \, \mathrm{d} y \\
&= f_{\mathrm{GIG}}(x \, ; \, p, a, b),
\end{align*}
as we wanted to show.
\end{proof}

\begin{prop} Let $G_p(x)$ be the CDF of the $\mathrm{GIG}(p,a,b)$. For half-integer $p$ such that $|p| > 1/2$, the following recurrence holds:
$$
G_p(x) = w G_{p-2}(x) + (1-w) G_{p-1}(x) - \exp\left\{ \frac{-ax}{2} \right\} I_p(x),
$$
where
$$
I_p(x) = \frac{w (ab)^{(p-2)/2}}{2^{p-1} K_{p-2}(\sqrt{ab})} \Gamma\left(2-p, \frac{b}{2x}\right) + \frac{(1-w) (ab)^{(p-1)/2}}{2^p K_{p-1}(\sqrt{ab})} \Gamma\left(1-p, \frac{b}{2x}\right)
$$
and $\Gamma( \cdot,  \cdot)$ is the incomplete gamma function  \citep{abramowitz1968handbook}.
\end{prop}

\begin{proof}
    Let $G_p$ be the CDF of $\mathrm{GIG}(p, a, b)$. Then,
\begin{align*}
    G_p(x) &=  \int_0^\infty  \left( \int_y^{\max(x,y)} \frac{a}{2} \exp\left\{ - \frac{a}{2} (x - y) \right\}  \, \mathrm{d} x  \right) \, f_Y(y)  \, \mathrm{d} y  \\
    &= \int_{0}^\infty \left[1- \exp\left\{\frac{a(y-\max(x,y))}{2} \right\} \right] \, f_Y(y) \, \mathrm{d} y  \\
    &= \int_0^x \left[1- \exp\left\{\frac{a(y-x)}{2} \right\} \right] \, f_Y(y) \, \mathrm{d} y \\
    &= \int_0^x f_Y(y) \, \mathrm{d} y - \exp\left(\frac{-ax}{2}\right) \int_0^x \exp\left(\frac{ay}{2}\right) f_Y(y) \, \mathrm{d} y.
\end{align*}
On the one hand,
$$
\int_0^x f_Y(y) \, \mathrm{d} y = w G_{p-2}(x) + (1-w) G_{p-1}(x),
$$
where $G_{p-2}$  and $G_{p-1}$ are the CDFs of $\mathrm{GIG}(p-2, a, b)$ and $\mathrm{GIG}(p-1, a, b)$, respectively. On the other hand,
\begin{align*}
 &\int_0^x \exp\left(\frac{ay}{2}\right) f_Y(y) \, \mathrm{d} y =  I_p(x) \\
 &I_p(x) = \frac{w (ab)^{(p-2)/2}}{2^{p-1} K_{p-2}(\sqrt{ab})} \Gamma\left(2-p, \frac{b}{2x}\right) + \frac{(1-w) (ab)^{(p-1)/2}}{2^p K_{p-1}(\sqrt{ab})} \Gamma\left(1-p, \frac{b}{2x}\right),
\end{align*}
where $\Gamma( \cdot,  \cdot)$ is the incomplete gamma function  \citep{abramowitz1968handbook}.
\end{proof}

\subsection{\textcolor{black}{Algorithms}}

\color{black}

In this section, we give pseudocode to implement the algorithms. In Algorithm 3 (CDF evaluation for half-integer $p$), we take a bottom-up approach (a \texttt{for} loop). This is more efficient than a recursive function because it avoids redundant computations.

\begin{algorithm}[H]
\caption{Data-augmented composition Gibbs sampler for $\mathrm{GIG}(p, a, b)$}\label{alg:gibbs}
\begin{algorithmic}
\Require $p \in \mathbb{R}, a > 0,  b > 0, n_{\mathrm{sim}} \in \{1, 2, \, ... \}$
\State $x \gets \text{vector}(\text{length} = n_{\mathrm{sim}})$
\If{$p < -1/2$}
\State $y \gets -(p+1/2)\sqrt{b/a}$ 
\For{$i$ in 1 to $n_\mathrm{sim}$}
\State $x[i] \sim \mathrm{InvGauss}(\sqrt{b/(a+2y)},  b)$
\State $y \sim \mathrm{Gamma}(-(p+1/2), x)$
\EndFor
\ElsIf{$p = -1/2$}
\For{$i$ in 1 to $n_\mathrm{sim}$}
\State $x[i] \sim \mathrm{InvGauss}(\sqrt{b/a}, a)$
\EndFor
\Else
\State $y \gets (p+1/2)\sqrt{b/a}$ 
\For{$i$ in 1 to $n_\mathrm{sim}$}
\State $x[i] \sim \mathrm{InvGauss}(\sqrt{(b+2y)/a},  b+2y)$
\State $y \sim \mathrm{Gamma}(p+1/2, 1/x)$
\EndFor
\EndIf
\State \Return $x$
\end{algorithmic}
\end{algorithm}

\clearpage

\begin{algorithm}[h]
\caption{rgig.half: Sampling from $\mathrm{GIG}(p, a, b)$ for half-integer $p$}\label{alg:rgig.half}
\begin{algorithmic}
\Require $p \in \{ \, ... \, , -3/2, -1/2, 1/2, 3/2, \, ... \}, a > 0,  b > 0$
\If{$p < 0$}
    \State $x \gets 1/\mathrm{rgig.half}(-p, b, a)$
\ElsIf{$p = 1/2$}
    \State $y \sim \mathrm{InvGauss}(\sqrt{a/b}, a)$
    \State $x \gets 1/y$
\Else
    \State $w \gets K_{p-2}(\sqrt{ab})/K_p(\sqrt{ab})$
    \State $U \sim \mathrm{Unif}(0, 1)$
    \State $E \sim \mathrm{Exp}(a/2)$
    \If{$U < w$}
        \State $x \gets \mathrm{rgig.half}(p-2, a, b) + E$
    \Else 
        \State $x \gets \mathrm{rgig.half}(p-1, a, b) + E$
    \EndIf
\EndIf
\State \Return $x$
\end{algorithmic}
\end{algorithm}

\clearpage

 \begin{algorithm}[h!]
\caption{pgig.half: CDF of $\mathrm{GIG}(p, a, b)$ for half integer $p$}\label{alg:CDF}
\begin{algorithmic}
\Require $x > 0, p \in \{ \, ... \, , -3/2, -1/2, 1/2, 3/2, \, ... \}, a > 0,  b > 0$
\If{$p <0$}
    \State $1-\mathrm{pgig.half}(1/x, -p, b, a)$
\Else
    \State $v \gets \mathrm{vector}(\text{length} = p+1/2)$
    \State $v[1] \gets P[\mathrm{InvGauss}(\sqrt{a/b}, b) > 1/x]$
    \If{$\text{length}(v) > 1$}
        \State $w \gets K_{-1/2}(\sqrt{ab})/K_{3/2}(\sqrt{ab})$
        \State $I \gets  \frac{w (ab)^{-1/4} \Gamma\left(1/2, \frac{b}{2x}\right)}{2^{1/2} K_{-1/2}(\sqrt{ab})}  + \frac{(1-w) (ab)^{1/4} \Gamma\left(-1/2, \frac{b}{2x}\right)}{2^{3/2} K_{1/2}(\sqrt{ab})} $
        \State $k \gets \exp\{-ax/2\}$
        \State $v[2] \gets w \, \text{pgig.half}(x, -1/2, a, b) + (1-w) \, v[1] - k I$
        \If{$\text{length}(v) > 2$}
            \For{$i \text{ in } 3:\text{length}(v)$}
                \State $w \gets K_{i-5/2}(\sqrt{ab})/K_{i-1/2}(\sqrt{ab})$
                \State $I \gets \frac{w (ab)^{(i-5/2)/2} \Gamma\left(5/2-i, \frac{b}{2x}\right)}{2^{i-3/2} K_{i-5/2}(\sqrt{ab})}  + \frac{(1-w) (ab)^{(i-3/2)/2} \Gamma\left(3/2-i, \frac{b}{2x}\right)}{2^{i-1/2} K_{i-3/2}(\sqrt{ab})} $
                \State $v[i] \gets w \, v[i-2] + (1-w) \, v[i-1] - k I$
            \EndFor 
        \EndIf
    \EndIf
    \State \Return $v[\text{length}(v)]$
\EndIf
\end{algorithmic}
\end{algorithm}

\clearpage

\section{Reparametrization of the GIG}

This section discusses a reparametrization of the GIG and restates all the propositions in the main text in this alternative parametrization.

\subsection{Reparametrization}

Let $f_X(x)$ denote the PDF of $X \sim \text{GIG}(p, a, b)$, given by:
$$
f_X(x) = \frac{\left(a/b\right)^{p/2}}{2 K_p(\sqrt{ab})} x^{p-1} \exp\left(-\frac{ax + b/x}{2}\right), \quad x > 0,
$$
where $K_p(\cdot)$ is the modified Bessel function of the second kind. The parameters are $a > 0$, $b > 0$, and $p \in \mathbb{R}$.

Consider the reparametrization
$$
\omega = \sqrt{ab}, \quad \eta = \sqrt{b/a}, \quad p = p.
$$
Substituting these into the PDF, the density function becomes:
$$
f_X(x) = \frac{1}{2 \eta^p K_p(\omega)} x^{p-1} \exp\left(-\frac{\omega}{2} \left(\frac{x}{\eta} + \frac{\eta}{x}\right)\right), \quad x > 0.
$$

This reparametrization is useful because generating a random draw from $\text{GIG}(p, a, b)$ is equivalent to generating from $\text{GIG}(p, \omega, \eta = 1)$ and then rescaling the draws:
\begin{enumerate}
    \item Compute $\omega = \sqrt{ab}$ and $\eta = \sqrt{b/a}$.
    \item Generate a random draw $Z \sim \text{GIG}(p, \omega, \eta = 1)$.
    \item Scale $Z$ by $\eta$ to obtain $X = \eta Z$. 
\end{enumerate}

\textbf{Proof:}  
Let $Z \sim \text{GIG}(p, \omega, \eta = 1)$ and define $X = \eta Z$, where $\eta = \sqrt{b/a}$. Using the change of variables $z = x / \eta$, the Jacobian is $\mathrm{d}z / \mathrm{d}x = 1 / \eta$. The PDF of $X$ is:
$$
f_X(x) = f_Z\left(\frac{x}{\eta}\right) \frac{1}{\eta}.
$$
Substituting $f_Z(z)$ and simplifying:
$$
f_X(x) = \frac{1}{2 K_p(\omega)} \left(\frac{x}{\eta}\right)^{p-1} \exp\left(-\frac{\omega}{2} \left(\frac{x}{\eta} + \frac{\eta}{x}\right)\right) \frac{1}{\eta}.
$$
Combine terms:
$$
f_X(x) = \frac{\eta^{-p}}{2 K_p(\omega)} x^{p-1} \exp\left(-\frac{ax + b/x}{2}\right).
$$
Since $\eta^{-p} = \left(a/b\right)^{p/2}$, we recover:
$$
f_X(x) = \frac{(a/b)^{p/2}}{2 K_p(\omega)} x^{p-1} \exp\left(-\frac{ax + b/x}{2}\right).
$$
Thus, $X \sim \text{GIG}(p, a, b)$.

\subsection{Propositions in the new parametrization}

This section restates the propositions using $\omega$ and $\eta$. Propositions 4 and 5 involve the weight $w$, which determines how different GIGs are combined. Unfortunately, the notation $w$ clashes with the new parameter $\omega$, so we redefine
$$
\tau = w = \frac{K_{p-2}(\omega)}{K_p(\omega)}.
$$

\vspace{1pc}
\hrule
\vspace{1pc}

\textbf{Proposition 1:}  
The PDF of $X \sim \text{GIG}(p, \omega, \eta)$ can be expressed as a mixture of inverse Gaussian distributions:
$$
f_X(x) = \int_0^\infty f_{X|Y}(x \mid y) f_Y(y) \, \mathrm{d}y,
$$
where:
\begin{itemize}
    \item For $p < -1/2$:
    $$
    f_{X|Y}(x \mid y) = f_{\mathrm{IG}}(x  \, ; \,  \sqrt{\omega \eta/(\omega + 2 \eta y)},  \omega \eta ),
    $$
    $$
    f_Y(y) = \frac{1}{\mathcal{Z}_0} \, y^{-(p+3/2)} \exp\left\{-\, \omega \sqrt{1+2\eta y/\omega}\right\},
    $$
    with 
    $$
    \mathcal{Z}_0 = \frac{\sqrt{2}\,\Gamma\bigl(-p-\tfrac{1}{2}\bigr) \, K_p(\omega) \, \omega^{1/2} \, \eta^{\,p+1/2}}{\sqrt{\pi}}.
    $$
       \item For $p > -1/2$:
    $$
    f_{X|Y}(x \mid y) = f_{\mathrm{IG}}(x \, ; \, \sqrt{\eta^2+2\eta y/\omega}, \omega \eta + 2y),
    $$
    $$
    f_Y(y) = \frac{1}{\mathcal{Z}_{1}} \, \frac{y^{p-1/2}}{\sqrt{\omega \eta + 2y}} \,  \exp \left\{ -\omega\sqrt{1 + 2  y /(\omega \eta)} \right\} ,
    $$
    with 
    $$
    \mathcal{Z}_1 = \frac{\sqrt{2}\,\Gamma\bigl(p+1/2) \, K_p(\omega) \, \eta^{p}}{\sqrt{\pi}}.
    $$

\end{itemize}

\vspace{1pc}
\hrule
\vspace{1pc}
  
\noindent The Gibbs sampler iteratively alternates between:
\begin{enumerate}
    \item Sample $Y \mid X = x$:
    \begin{itemize}
        \item For $p > -1/2$, $Y \mid X = x \sim \text{Gamma}(p + 1/2, 1/x)$.
        \item For $p < -1/2$, $Y \mid X = x \sim \text{Gamma}(-(p + 1/2), x)$.
    \end{itemize}
    \item Sample $X \mid Y = y$:
    \begin{itemize}
        \item For $p < -1/2$, $X \mid Y = y \sim\mathrm{InvGauss}(\sqrt{\omega \eta/(\omega + 2 \eta y)},  \omega \eta )$.
        \item For $p > -1/2$, $X \mid Y = y \sim \mathrm{InvGauss}(\sqrt{\eta^2+2\eta y/\omega}, \omega \eta + 2y)$.

    \end{itemize}
\end{enumerate}

\vspace{1pc}
\hrule
\vspace{1pc}

\textbf{Proposition 4:}  
If $X \sim \text{GIG}(p, \omega, \eta)$ and $p > 1$, then:
$$
X =_d Y + E,
$$
where $Y$ and $E$ are independent random variables such that
$$
\color{black}
f_Y(y) = \tau \, f_{\text{GIG}}(y \, ; \, p-2, \omega, \eta) + (1-\tau) \, f_{\text{GIG}}(y \, ; \, p-1, \omega, \eta)
$$
and $E \sim \text{Exp}(\omega /(2\eta))$.
\vspace{1pc}
\hrule
\vspace{1pc}

\textbf{Proposition 5:}  
\noindent
Let $G_p(x)$ be the CDF of the $\mathrm{GIG}(p,\omega,\eta)$. For half-integer $p$ such that $|p| > 1/2$, the following recurrence holds:
$$
G_p(x) = \tau G_{p-2}(x) + (1-\tau) G_{p-1}(x) - \exp\left(-\frac{\omega x}{2\eta}\right) I_p(x),
$$
where
$$
I_p(x) = \frac{\tau \,\omega^{\,p-2}}{2^{\,p-1} K_{p-2}(\omega)} \Gamma\left(2-p, \frac{\omega \eta}{2x}\right)
\;+\; \frac{(1-\tau)\,\omega^{\,p-1}}{2^{\,p} K_{p-1}(\omega)} \Gamma\left(1-p, \frac{\omega \eta}{2x}\right),
$$
and $\Gamma(\cdot,\cdot)$ is the incomplete gamma function.
\vspace{1pc}
\hrule
\vspace{1pc}

\subsection{Gibbs sampler based on Zhang and Reiter (2022)}

This section describes an alternative Gibbs sampler for GIG that uses a decomposition presented in \cite{zhang2022generator}. The Gibbs sampler iterates between sampling from truncated distributions and the support of the conditionals changes from iteration to iteration. This characteristic can introduce challenges in verifying certain convergence properties, as discussed later.

\subsubsection{Notation for Truncated Distributions}
We denote truncated distributions using the following notation:
\begin{itemize}
    \item $X \sim \text{TInvGamma}(\alpha, \beta, L, U)$: a truncated inverse gamma distribution with shape parameter $\alpha$, scale parameter $\beta$, and truncation bounds $(L, U)$.
    \item $Y \sim \text{TExp}(\lambda, L, \infty)$: a truncated exponential distribution with rate $\lambda$ and lower truncation bound $L$.
\end{itemize}

Truncated densities are normalized over their truncation bounds.

\subsubsection{Joint Distribution}

Using results in \cite{zhang2012}, we can write a joint distribution of $X$ and $Y$ for $X \sim \text{GIG}(p, a, b)$:
$$
f_{X,Y}(x, y) \propto \exp\left(-\frac{b y}{2}\right) \, x^{p-1} \exp\left(-\frac{a}{2x}\right), \quad 0 < x < y,
$$
where $p \le 0$. 

From this joint distribution, the two full conditional distributions are derived by isolating terms involving $x$ and $y$, respectively.

\subsubsection{Full Conditional Distributions}

From $f_{X,Y}(x, y)$, the terms involving $x$ give:
$$
f_{X \mid Y}(x \mid y) \propto x^{p-1} \exp\left(-\frac{a}{2x}\right), \quad 0 < x < y.
$$
This corresponds to:
$$
X \mid Y = y \sim \text{TInvGamma}(-p, a/2, 0, y).
$$
Similarly, the terms involving $y$ yield:
$$
f_{Y \mid X}(y \mid x) \propto \exp\left(-\frac{b y}{2}\right), \quad y > x.
$$
This corresponds to:
$$
Y \mid X = x \sim \text{TExp}(b/2, x, \infty).
$$

\subsubsection{Generalization for all $p$}

For $p > 0$, we can use the property $1/X \sim \text{GIG}(-p, b, a)$ to extend the sampler:
\begin{enumerate}
    \item Transform $X$ to $Z = 1/X$ and apply the Gibbs sampler for $Z \sim \text{GIG}(-p, b, a)$.
    \item After sampling $Z$, invert to obtain $X = 1/Z$.
\end{enumerate}

\subsubsection{Assessing convergence}

Theorem 1 in \cite{johnson2015geometric} provides sufficient conditions for geometric ergodicity of two-stage Gibbs samplers. In this case, the support of the transition kernel changes at every iteration, which complicates the verification of Assumption $\mathcal{A}$ in \cite{johnson2015geometric}. It is known that changing the support of the transition kernel from iteration to iteration can negatively affect the properties of the sampler \citep{robert2004monte, meyn2012markov}.

\subsection{Numerical experiments}

We conducted two experiments to compare the efficiency of the data-augmented Gibbs sampler and the exact simulation method with the methods proposed by \citet{devroye2014random}, as implemented in the \texttt{R} package \texttt{boodist} \citep{boodist}, and \citet{hormann2014generating}, as implemented in the \texttt{R} package \texttt{GIGrvg} \citep{GIGrvg}. We compared with the methods from these articles because, as far as we are aware, they represent the two most recent published works on simulating from the generalized inverse Gaussian distribution. [In the more recent preprint of \citet{zhang2022generator}, the authors mention that their method never outperformed that of \citet{hormann2014generating}.] Before we describe the experiments, it is important to note that the efficiency of a method is highly dependent upon the details of its implementation, including the programming language, the skill of the programmer, etc. 

In the first experiment, we compared the two Gibbs samplers described in the example toward the end of Section 2. The original Gibbs sampler alternates between normal and GIG full conditionals, while the data-augmented Gibbs sampler cycles through normal, inverse Gaussian, and gamma full conditionals. We considered two versions of the original Gibbs sampler: one using the method of \citet{devroye2014random} implemented in \texttt{boodist} and one using the method of \citet{hormann2014generating} implemented in \texttt{GIGrvg}. The data were simulated from the normal distribution with parameters $\mu=1, \sigma^2=1$ and sample size $n=100.$ The prior hyperparameters were chosen as $\theta_0 = 0, \tau_0^2=100, p_0=3/4, a_0=1, b_0=1.$  We ran the Gibbs samplers for 5000 iterations. Based on 1000 repetitions, we found that the data-augmented Gibbs sampler was about 8\% slower than the version of the original Gibbs sampler using the method of \citet{hormann2014generating} implemented in \texttt{GIGrvg}. However, the data-augmented Gibbs sampler was more than 9 times faster than the version of the original Gibbs sampler using the method of \citet{devroye2014random} implemented in \texttt{boodist}. Again, these differences may tell us more about the quality of the implementations of each method than they tell us about the methods themselves. 

Neither the original Gibbs sampler nor the data-augmented Gibbs sampler produce independent draws from the posterior distribution, so it is important to evaluate the effective sample size (ESS) in addition to the time it takes to produce the Markov chains. For the original Gibbs sampler, we found that the ESS was close to the nominal sample size for both parameters $\mu$ and $\sigma^2.$ For the data-augmented Gibbs sampler, this was only true for $\mu.$ The ESS of $\sigma^2$ was about 1/3 of the nominal sample size. \textcolor{black}{To assess the ESS, we simulated 100 chains of length 50000 with 5000 warmup iterations and used the \texttt{effectiveSize} function from the \texttt{R} package \texttt{coda} \citep{coda}.}

The conclusions of the previous two paragraphs did not seem to be sensitive to the particular choices of parameters, sample size, or prior hyperparameters. \textcolor{black}{A summary of the results of the first simulation experiment appears in Table \ref{table:exp1}.}

\begin{table}[h!]
\caption{\color{black} A summary of the results of the first simulation experiment. The ``Time" column provides the cumulative time in seconds to produce 1000 chains of length 5000. The ``Relative" column is computed by dividing the ``Time" column by its smallest value. The last two columns report the median ESS per iteration of $\mu$ and $\sigma^2$ from 100 chains of length 50000 with 5000 warmup iterations. We do not report ESS per iteration for the method of \citet{devroye2014random} because it should not differ from that of the method of \citet{hormann2014generating}.}
\centering
\color{black}
\begin{tabular}{||c c c c c||} 
 \hline
 Sampler & Time & Relative & ESS/iter $\mu$ &  ESS/iter $\sigma^2$  \\ [0.5ex] 
 \hline\hline
 Proposed & 26.564 & 1.085 & 1.00 & .339 \\ 
 L \& H & 24.491 & 1.000 & 1.00 & .981 \\
 Devroye & 228.304 & 9.322 & * & * \\ [1ex] 
 \hline
\end{tabular}
\label{table:exp1}
\end{table}

\color{black}

In the second experiment, we compared the exact simulation method of Section 3 with the methods of \citet{devroye2014random} and \citet{hormann2014generating} in the case where $p=3/2.$ The purpose of the experiment was to establish that the exact simulation method outperforms the other methods when $|p|$ is sufficiently small. We set $p=3/2, a=1, b=1.$ Based on $10000$ repetitions, we found that the method of \citet{hormann2014generating} implemented in \texttt{GIGrvg} took almost 10\% longer than the exact simulation method of Section 3. The method of \citet{devroye2014random} implemented in \texttt{boodist} took about 87\% longer than the exact simulation method. These conclusions did not seem to be sensitive to the particular choices of the parameters $a$ and $b.$ \textcolor{black}{A summary of the results of the second simulation experiment appears in Table \ref{table:exp2}.}

\begin{table}[h!]
\caption{\textcolor{black}{A summary of the results of the second simulation experiment. The ``Time" column provides the cumulative time in seconds to complete all 10000 repetitions. The ``Relative" column is computed by dividing the ``Time" column by its smallest value.}}
\centering
\color{black}
\begin{tabular}{||c c c||} 
 \hline
 Sampler & Time & Relative \\ [0.5ex] 
 \hline\hline
 Proposed & 8.497 & 1.000  \\ 
 L \& H & 9.317 & 1.097  \\
 Devroye & 15.858 & 1.866  \\ [1ex] 
 \hline
\end{tabular}
\label{table:exp2}
\end{table}

In our view, the data-augmented Gibbs sampler and the exact simulation method complement the existing methods for simulating from the GIG. In the first experiment, we saw that our implementation of the data-augmented Gibbs sampler can be more efficient than the \texttt{R} package \texttt{boodist} for simulating from the GIG. In the second experiment, we saw that our implementation of the exact simulation method can be more efficient than both \texttt{boodist} and \texttt{GIGrvg} when $|p|$ is sufficiently small. Code to replicate these experiments is available at \url{https://github.com/michaeljauch/gig}.

\clearpage

\color{black}

\bibliographystyle{chicago} 
\bibliography{GIGs}

\end{document}